\documentclass[10pt, final, conference, letterpaper, onecolumn, oneside]{./IEEEtran}

\usepackage{cite}
\usepackage{graphicx}
\usepackage{subfigure}
\usepackage{amsmath}
\usepackage{amssymb}
\usepackage{epsfig}
\usepackage{times}
\usepackage{color}
\usepackage{yhmath}
\usepackage{./clrscode_mod}

\newtheorem{pavikc}{\textbf{Corollary}}
\newtheorem{pavikl}{\textbf{Lemma}}
\newtheorem{pavikt}{\textbf{Theorem}}

\newcommand{\argmin}{\operatornamewithlimits{argmin}}
\newcommand{\argmax}{\operatornamewithlimits{argmax}}

\begin{document}

\title{Distributed Function Computation in Asymmetric Communication Scenarios}%
\author{\IEEEauthorblockN{Samar Agnihotri\IEEEauthorrefmark{2} and Rajesh Venkatachalapathy\IEEEauthorrefmark{4}}%
\IEEEauthorblockA{\IEEEauthorrefmark{2}CEDT, Indian Institute of Science, Bangalore - 560012, India\\}
\IEEEauthorblockA{\IEEEauthorrefmark{4}Systems Science Graduate Program, Portland State University, Portland, OR 97207\\}
Email: samar@cedt.iisc.ernet.in, venkatr@pdx.edu%
}

\IEEEspecialpapernotice{(Extended Abstract)}

\maketitle

\begin{abstract}
We consider the distributed function computation problem in asymmetric communication scenarios, where the sink computes some deterministic function of the data split among $N$ correlated informants. The distributed function computation problem is addressed as a generalization of distributed source coding (DSC) problem. We are mainly interested in minimizing the number of informant bits required, \textit{in the worst-case}, to allow the sink to exactly compute the function. We provide a constructive solution for this in terms of an interactive communication protocol and prove its optimality. The proposed protocol also allows us to compute the worst-case achievable rate-region for the computation of any function. We define two classes of functions: \textit{lossy} and \textit{lossless}. We show that, in general, the \textit{lossy} functions can be computed at the sink with fewer number of informant bits than the DSC problem, while computation of the \textit{lossless} functions requires as many informant bits as the DSC problem.
\end{abstract}

\section{Introduction}
\label{sec:Intro}
Let us consider a distributed function computation scenario, where a sink node is interested in \textit{exactly} computing some deterministic function $f = f(\overline{X})$ of data-vector $\overline{X}$ that is split among $N$ correlated informants. The correlation in informants' data is modeled by discrete and finite distribution $\cal P$, known only to the sink (asymmetric communication, \cite{097kushilevitzNisan}). The sink and informants interactively communicate with each other, with communication proceeding in rounds, as in \cite{079yao}. We are concerned with minimizing the number of bits that the informants send, \textit{in the worst-case}, to allow the sink to compute the function.

We consider the distributed function computation problem as a generalization of \textit{distributed source coding} (DSC) problem\footnote{DSC problem is a special case of distributed function computation problem where the function to be computed is identity map, ${id}_{\overline{X}}$.}. The particular distributed function computation problem we consider is a generalization of DSC problem in asymmetric communication scenarios, we addressed in \cite{allerton08}. As for that work, the motivation for this work too comes from sensor networks, particularly from our efforts to address the distributed function computation problem in single-hop data-gathering wireless sensor networks, while maximizing the worst-case operational lifetime of the network. In a typical data-gathering sensor network, it is reasonable to assume that the base-station has large resources of energy, computation, and communication as well as the knowledge of correlations in sensor data, whereas a sensor node is resource limited and only knows its sampled data-values. Therefore, we argue that in such communication scenarios, the onus should be on the base-station to bear most of the burden of computation and communication associated with function computation. Allowing interactive communication between the base-station and sensor nodes lets us precisely do this: base-station forms and communicates \textit{efficient} queries to sensor nodes, which they respond to with short and easily computable messages. This reduces the communication and computation effort at sensor nodes, hence enhancing their lifetime, which in turn leads to increased network lifetime.

The distributed function computation problem was first addressed by Yao in \cite{079yao} and later by other researchers in different setups, as we discuss in Section~\ref{sec:relatedWork}. However, our work mainly differs from the extant work in one or more aspects as follows. First, we approach the distributed function computation problem as a generalization of DSC problem. This allows us to exploit the correlation in informants' data to solve the function computation problem at the sink with fewer informant bits. Second, we are concerned with asymmetric communication (only sink knows the joint distribution of informants' data) and asymmetric computation (only sink computes the function). Third, we are concerned with the worst-case analysis. Fourth, we are interested in distributed function computation with a single instance of data at informants (\textit{one-shot} computation problem). Finally, we consider a more powerful model of communication where the sink and informants interactively communicate with each other. Our work allows us to clearly delineate the roles played in optimally solving distributed function computation problem in arbitrary networks by correlation in informants' data, the properties of the function to be computed, communication model, network connectivity graph, and routing strategies. In this sense, our work acts as a fundamental building block to a general theory of distributed function computation over arbitrary networks, which we expect to eventually develop.

In Section~\ref{sec:ambiguity}, we revisit the notion of \textit{information ambiguity}, an information measure we proposed in \cite{isita08} for the worst-case information-theoretic analyses, and extend it to a form useful in the present context. In Section~\ref{sec:probSetting}, we provide the details of the communication model we assume and formally introduce the variant of distributed function computation problem we address in this paper. In the next section, we give a communication protocol to compute any given function at the sink, prove its optimality with respect to minimizing the number of informant bits, and provide the bounds on its performance. Finally in Section~\ref{sec:functionProperties}, we discuss some properties of distributed function computation problem and propose a classification scheme for functions, based on the number of informant bits required, in general, to compute those at the sink.

\section{Related Work}
\label{sec:relatedWork}
There are three major existing approaches to address the distributed function computation problem, as follows:

\textit{Communication complexity:} The seminal paper by Yao \cite{079yao} introduced the problem of computing the minimum number of bits exchanged between two processors when both the processors compute a function of the input that is split between processors. Variants of this problem and numerous solution approaches have been explored in the field of communication complexity, \cite{097kushilevitzNisan}. This work provides insights into developing efficient communication protocols for function computation. However, it is mainly interested in estimating the order-of-magnitude of the bounds on communication and computation costs. Also, it is not obvious how to extend this work, when for example, one or more nodes in the network are interested in computing some function of source nodes' data or the source data is split among more than two nodes and is possibly correlated.

\textit{Scaling laws:} Recently in \cite{105giridharKumar, 108khudeKarnikKumar, 108kamathManjunath}, the distributed function computation problem has been addressed to find how the rate of function computation scales with network size. This approach however does not provide a simple framework to exploit the correlation in source data and to incorporate stronger models of computation and communication, such as interactive communication, data-buffers, cooperating sources.

\textit{Information theory:} Much before Yao introduced his formulation of distributed function computation problem, Slepian and Wolf in \cite{073slepianWolf} introduced the DSC problem. It was many years before distributed function computation problem was seriously addressed in information-theoretic setup, \cite{087hanKobayashi, 088gallager, 096schulman, 101orlitskyRoche}. Still, there is very little such work that comprehensively addresses the distributed function computation problem over any given network, function, and model of communication and computation.

\section{Information Ambiguity for Distributed Function Computation}
\label{sec:ambiguity}
We revise and generalize some relevant definitions and properties of \textit{information ambiguity}, an information measure we introduced in \cite{isita08} for performing the worst-case information-theoretic analysis in certain communication scenarios. We then extend the notion of information ambiguity to a form useful for distributed function computation in this paper.

\textit{Note:} All the logarithms used in this paper are to the base 2, unless explicitly mentioned otherwise.

Let us consider a $N$-tuple of random variables $(X_1, \ldots, X_N) \sim {\cal P} = p(x_1, \ldots, x_N), X_i \in {\cal X}, i \in \{1, \ldots, N\}$, where ${\cal X}$ is discrete and finite alphabet of size $|\cal X|$. The \textit{support set} of $(X_1, \ldots, X_N)$ is defined as:
\begin{equation}
\label{eqn:ambiguity_set_dstrbnN}
S_{X_1, \ldots, X_N} \stackrel{\textrm{def}}{=} \{(x_1, \ldots, x_N) | p(x_1, \ldots, x_N) > 0 \}
\end{equation}
We also call $S_{X_1, \ldots, X_N}$ as the \textit{ambiguity set} of $(X_1, \ldots, X_N)$. The cardinality of $S_{X_1, \ldots, X_N}$ is called \textit{ambiguity} of $(X_1, \ldots, X_N)$ and denoted as $\mu_{X_1, \ldots, X_N} = |S_{X_1, \ldots, X_N}|$. So, the minimum number of bits required to describe an element of $S_{X_1, \ldots, X_N}$, in the worst-case, is $\lceil \log \mu_{X_1, \ldots, X_N} \rceil$.

The \textit{support set} $S_{X_i}$ of $X_i, i \in \{1, \ldots, N\}$, is the set
\begin{equation}
\label{eqn:ambiguity_set_varN}
S_{X_i} \stackrel{\textrm{def}}{=} \{x_i: \mbox{ for some } x_{-i}, (x_{-i}, x_i) \in S_{X_1, \ldots, X_N}\}, \mbox{with } x_{-i} \stackrel{\textrm{def}}{=} \{x_1, \ldots, x_N\} \setminus x_i
\end{equation}
of all possible $X_i$ values. We also call $S_{X_i}$ \textit{ambiguity set} of $X_i$. The \textit{ambiguity} of $X_i$ is defined as $\mu_{X_i} = |S_{X_i}|$. The \textit{conditional ambiguity set} of $(X_1, \ldots, X_N)$, when random variable $X_i$ takes the value $x_i, x_i \in S_{X_i}$, is
\begin{equation}
\label{eqn:cond_ambiguity_setN}
S_{X_1, \ldots, X_N|X_i}(x_i) \stackrel{\textrm{def}}{=} \{(x_1, \ldots, x_N): (x_1, \ldots, x_N) \in S_{X_1, \ldots, X_N} \mbox{ and } x_i \in S_{X_i}\},
\end{equation}
the set of possible $(X_1, \ldots, X_N)$ values when $X_i = x_i$. The \textit{conditional ambiguity} in that case is $\mu_{X_1, \ldots, X_N|X_i}(x_i) = |S_{X_1, \ldots, X_N|X_i}(x_i)|$, the number of possible values of $(X_1, \ldots, X_N)$ when $X_i =x_i$. The \textit{maximum conditional ambiguity} of $(X_1, \ldots, X_N)$ is
\begin{equation}
\label{eqn:max_cond_ambiguityN}
\widehat{\mu}_{X_1, \ldots, X_N|X_i} \stackrel{\textrm{def}}{=} \sup \{\mu_{X_1, \ldots, X_N|X_i}(x_i): x_i \in S_{X_i}\},
\end{equation}
the maximum number of $(X_1, \ldots, X_N)$ values possible with any value that $X_i$ can take.


In fact, for any two subsets $X_A$ and $X_B$ of $\{X_1, \ldots, X_N\}$, such that $X_A \cup X_B \subseteq \{X_1, \ldots, X_N\}$ and $X_A \cap X_B = \phi$, we can define for example, \textit{ambiguity set} $S_{X_A}$ of $X_A$, \textit{conditional ambiguity set} $S_{X_A|X_B}(x_B)$ of $X_A$ given the set $x_B$ of values that $X_B$ can take, and \textit{maximum conditional ambiguity set} $S_{X_A|X_B}$ of $X_A$ for any set of values that $X_B$ can take, with corresponding \textit{ambiguity}, \textit{conditional ambiguity}, and \textit{maximum conditional ambiguity} given by $\mu_{X_A}$, $\mu_{X_A|X_B}(x_B)$, and $\widehat{\mu}_{X_A|X_B}$, respectively. However, for the sake of brevity, we do not develop the precise definitions of these quantities here.

Further, let us represent each of $\mu_{X_i}$ values that random variable $X_i$ can take in $\lceil \log \mu_{X_i} \rceil$ bits as $b^i_1 \ldots b^i_{\lceil \log \mu_{X_i} \rceil}$. Let $\mbox{binary}_j(x_i)$ represent the value of $j^{\textrm{th}}, 1 \le j \le \lceil \log \mu_{X_i} \rceil$, bit-location in the bit-representation of $x_i$. Then, knowing that the value of $j^{\textrm{th}}$ bit-location is $b, b \in \{0, 1\}$, we can define the set of possible values that $X_i$ can take as
\begin{equation}
\label{eqn:cond_ambiguity4Bit}
S_{X_i|b^i_j}(b) \stackrel{\textrm{def}}{=} \{x_i: x_i \in S_{X_i} \mbox{ and } \mbox{binary}_j(x_i) = b\},
\end{equation}
with corresponding cardinality denoted as $\mu_{X_i|b^i_j}(b)$. We can similarly define $S_{X_A|b^i_j}(b)$ with $X_i \in X_A$ as
\begin{equation}
\label{eqn:cond_ambiguityOset4Bit}
S_{X_A|b^i_j}(b) \stackrel{\textrm{def}}{=} \{x_A: x_A \in S_{X_A} \mbox{ and } \mbox{binary}_j(x_i) = b\},
\end{equation}
with corresponding cardinality denoted as $\mu_{X_A|b^i_j}(b)$. The definitions of conditional ambiguity sets in \eqref{eqn:cond_ambiguity4Bit} and \eqref{eqn:cond_ambiguityOset4Bit} can be easily extended to the situations where the values of one or more bit-locations in one or more random variable's bit-representation are known, but once more for the sake of brevity, we omit the details of such extended definitions.

Next, we introduce the notion of the ambiguity set and ambiguity of the function output values. The support-set of output values of some function $f$, also called \textit{ambiguity set} of function output values of function $f$, is defined as:
\begin{equation}
\label{eqn:ambiguity_set_f}
S_f \stackrel{\textrm{def}}{=} \{f(x_1,\ldots, x_N) : \mbox{ for some } (x_1,\ldots, x_N) \in S_{X_1,\ldots, X_N} \}
\end{equation}
The cardinality of $S_f$ is called \textit{ambiguity} of output values of function $f$ and denoted as $\mu_f = |S_f|$. So, the minimum number of bits required to describe an element in $S_f$ is $\lceil \log \mu_f \rceil$. The \textit{conditional ambiguity set} of function output values when $X_i = x_i, x_i \in S_{X_i}, i \in \{1, \ldots, N\}$, is defined as
\begin{equation}
\label{eqn:cond_ambiguity_set_f}
S_{f|X_i}(x_i) \stackrel{\textrm{def}}{=} \{f(x_1,\ldots, x_N) : \mbox{ for some } (x_1,\ldots, x_N) \in S_{{X_1,\ldots, X_N}|X_i} (x_i) \}
\end{equation}
The corresponding cardinality is called \textit{conditional ambiguity} of function output values when $X_i = x_i$ and denoted as $\mu_{f|X_i}(x_i)$. We can further define the \textit{maximum conditional ambiguity} of function output values as
\begin{equation}
\label{eqn:max_cond_ambiguity_f}
\widehat{\mu}_{f|X_i} \stackrel{\textrm{def}}{=} \sup \{\mu_{f|X_i}(x_i) : x_i \in S_{X_i} \}
\end{equation}
maximum number of function output values possible over any value that $X_i$ can take over $S_{X_i}$. The definitions in \eqref{eqn:cond_ambiguity_set_f} and \eqref{eqn:max_cond_ambiguity_f} can be similarly extended to the situations where the conditioning is carried out over a subset $X_A$ of $\{X_1, \ldots, X_N\}$. We omit the discussion of such extensions here.

Further, when the value of $j^{\textrm{th}}$ bit-location in the binary-representation of $x_i, x_i \in S_{X_i}$, is known, that is $b^i_j = b, b \in \{0, 1\}$, we can define corresponding conditional ambiguity set of function output values as follows
\begin{equation}
\label{eqn:cond_ambiguity_set_f4Bit}
S_{f|b^i_j}(b) \stackrel{\textrm{def}}{=} \{f(x_1, \ldots, x_N): (x_1, \ldots, x_N) \in S_{X_1, \ldots, X_N} \mbox{ and } x_i \in S_{X_i|b^i_j}(b)\},
\end{equation}
with corresponding cardinality denoted as $\mu_{f|b^i_j}(b)$.

If the function $f$ is defined for every $X_A, X_A \subset \{X_1, \ldots, X_N\}$, then for a given support-set $S_{X_1, \ldots, X_N}$ of data-vectors, the functional $\lceil \log \mu_f \rceil$ is a valid information measure as it satisfies various axioms of such measures, such as expansibility, monotonicity, symmetry, subadditivity, and additivity, \cite{106klir_book}. We omit the details of proof for the sake of brevity .

In the Table~\ref{table:notation}, we summarize the notation used frequently in this section and in the rest of the paper.

\begin{table}[!t]
\centering
\caption{Notation used frequently in the paper}
\vspace{-0.1in}
\begin{tabular}{c | l }
\hline
$N$ & number of informants\\
$\cal X$ & discrete and finite alphabet set of cardinality $|{\cal X}|$\\
${\cal P}$ & $N$-dimensional discrete probability distribution, ${\cal P} = p(x_1, \ldots, x_N), x_i \in {\cal X}$\\
$X_i$ & random variable observed by $i^{\textrm{th}}$ informant. $X_i \in {\cal X}$\\
$S_{X_i}$ & \textit{ambiguity set} at the sink of $i^{\textrm{th}}$ informant's data, with corresponding \textit{ambiguity} $\mu_{X_i} = |S_{X_i}|$\\
$S_{X_1, \ldots, X_N}$ & \textit{ambiguity set} at the sink of all informants' data, with corresponding \textit{ambiguity} $\mu_{X_1, \ldots, X_N} = |S_{X_1, \ldots, X_N}|$\\
$S_{{X_1, \ldots, X_N}|I}$ & \textit{conditional ambiguity set} at the sink of all informant's data, when sink has information $I$, with corresponding \\
 & \textit{conditional ambiguity} $\mu_{X_{i}|I} = |S_{X_{i}|I}|$. The exact nature of $I$ will be obvious from the context\\
$S_f$ & \textit{ambiguity set} at the sink of the output values of function $f$, with corresponding \textit{ambiguity} $\mu_f = |S_f|$\\
$S_{f|I}$ & \textit{conditional ambiguity set} at the sink of the output values of function $f$ when sink has information $I$, with corresponding\\
 & \textit{conditional ambiguity} $\mu_{f|I} = |S_{f|I}|$\\
$\#_f$ & minimum number of informant bits required in the worst-case to compute the function $f$ at the sink\\
$\#_{DSC}$ & minimum number of informant bits required in the worst-case to solve the DSC problem at the sink\\\hline
\end{tabular}
\label{table:notation}
\vspace{-0.2in}
\end{table}

\section{Distributed Function Computation in Asymmetric Communication Scenarios}
\label{sec:probSetting}
Let us consider a distributed function computation scenario, where a sink computes some function of the data of $N$ correlated informants. We assume the \textit{asymmetric communication}, where the joint distribution ${\cal P}$ of informants' data is known only to the sink. The Figure~\ref{fig:twoInformants} depicts this scenario for $N=2$.

\begin{figure}[!t]
\centering
\includegraphics[width=4.0in]{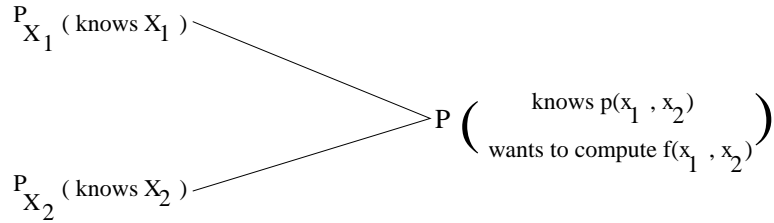}
\caption{Distributed function computation problem for two informants in asymmetric communication scenarios.}
\label{fig:twoInformants}
\vspace{-0.2in}
\end{figure}

\textbf{Problem Statement:} A sample $\overline{X} = (x_1, \ldots, x_N)$ is drawn \textit{i.i.d.} from a discrete and finite distribution $\cal P$ over $N$ binary strings, as in \cite{104chouPetrovic, 105adler}. The strings of $\overline{X}$ are revealed to $N$ informants, with the string $x_i, i \in \{1, \ldots, N\}$, being given to the $i^\textrm{th}$ informant. The sink wants to \textit{exactly} compute a deterministic function $f = f(\overline{X})$ of informants' data $\overline{X}$ (\textit{one-shot} computation problem). Our objective is to minimize the total number of informant bits required, \textit{in the worst-case}, to accomplish this.

\textbf{The Problem Setting:} We consider an \textit{asymmetric communication} scenario \cite{097kushilevitzNisan}. Communication takes place over $N$ binary, error-free channels, where each channel connects an informant with the sink. An informant and the sink can interactively communicate over the channel connecting them by exchanging messages (finite sequences of bits determined by agreed upon, deterministic protocol). The informants cannot communicate directly with each other, though. We assume that the communication between the sink and the informants proceeds in rounds, as in \cite{079yao}. In each round, depending on the information held by the communicators, one or other communicator may send the first message. However, we assume, as in \cite{090orlitsky}, that in each communication round, first the sink communicates to the informants and then, the informants respond with their messages. Each bit communicated over any channel is counted, as either a \textit{sink bit} if sent by the sink or an \textit{informant bit} if sent by an informant.

We assume the informants to be memoryless in the sense that they do not remember the messages they send in different rounds. We assume that $i^\textrm{th}$ informant knows its support-set $S_{X_i}$, so that it represents the binary string $x_i$, given to it, as $b_1^i \ldots b_{\lceil \log \mu_{X_i} \rceil}^i$ in $\lceil \log \mu_{X_i} \rceil$ bits.

The sink knows the distribution $\cal P$ and the corresponding support-sets: $S_{X_1, \ldots, X_N}$ of data-vectors and $S_f$ of function output values. So, every $\overline{X}, \overline{X} \in S_{X_1, \ldots, X_N}$, can be uniquely described using $\lceil \log \mu_{X_1, \ldots, X_N} \rceil$ bits and every $f(\overline{X})$ can be uniquely described using $\lceil \log \mu_f \rceil$ bits. This implies that to compute $f(\overline{X})$ unambiguously, the sink must receive at least $\lceil \log \mu_f \rceil$ bits from the informants, in the worst-case.

For the design and analysis of efficient communication protocols for distributed function computation, we develop a problem-encoding scheme as follows. Every informant data-vector $\overline{X}, \overline{X} \in S_{X_1, \ldots, X_N}$, can also be uniquely described by concatenating the bit-representations of all corresponding $x_i, 1 \le i \le N$. That is, $\overline{X}$ can be represented at the receiver by $\sum_{i = 1}^N \lceil \log \mu_{X_i} \rceil$ bits long representation, constructed by concatenating $i^\textrm{th}$ informant's $\lceil \log \mu_{X_i} \rceil$ bit-representation of $x_i$, for each $i \in \{1, \ldots, N\}$. With this encoding scheme, our distributed function computation problem reduces to minimizing the number $\#_f$ of bit-locations in the concatenated bit-representation of $\overline{X}$, whose values the sink needs to exactly compute $f(\overline{X})$. It should be noted that trivially, $\lceil \log \mu_f \rceil \le \#_f \le \sum_{i = 1}^N \lceil \log \mu_{X_i} \rceil$.

We illustrate this problem-encoding scheme with an example support-set in Figure~\ref{fig:probExa_f}. Let the informants 1 and 2 observe two correlated random variables $X_1$ and $X_2$, respectively, with $(X_1, X_2)$ derived from the support-set in first column. Let the function $f$ being computed at the sink be `bitwise OR' of the instance of $(X_1, X_2)$ revealed to the informants. For the given support-set, at least $\lceil \log \mu_{X_1, X_2} \rceil = 4$ bits are required to describe any element of $S_{X_1, X_2}$ and at least $\lceil \log \mu_f \rceil = 3$ bits are required to describe any element of $S_f$. Also, to individually describe any value assumed by $X_1$ and $X_2$, it requires $3$ bits.

For any given support-set of data-vectors, sink that knows the joint distribution $\cal P$, can construct a problem-encoding as in Figure~\ref{fig:probExa_f}. It knows that one string, hitherto unknown, from the fourth column is drawn, with first $\lceil \log \mu_{X_1} \rceil$ bits given to informant 1, next $\lceil \log \mu_{X_2} \rceil$ bits given to informant 2, and so on. We require the sink to exactly evaluate the given function $f$ on this string, whose different parts are held by different informants, with the informants sending minimum total number of bits to the sink.

\begin{figure*}[!t]
\centering
\includegraphics[width=7.0in]{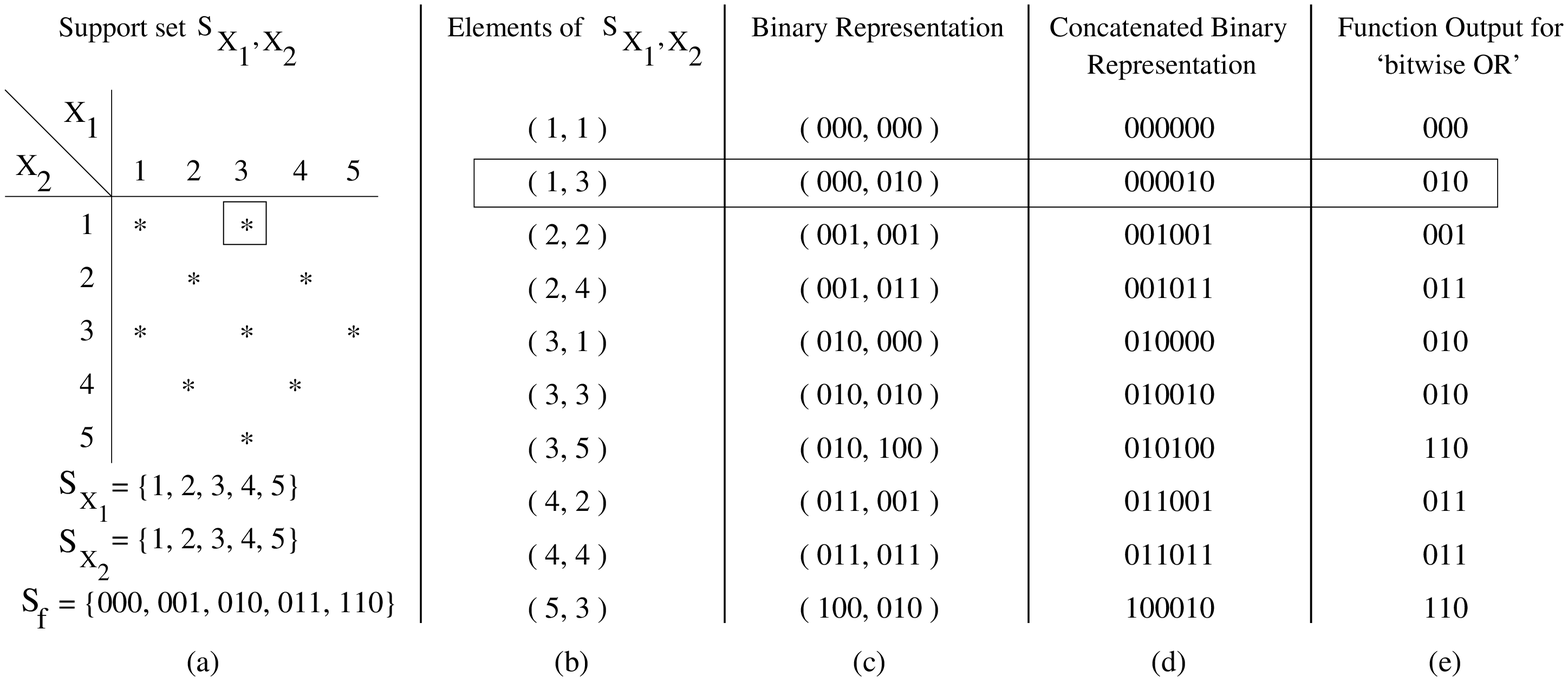}
\vspace{-0.1in}
\caption{Example of problem encoding: (a) Support-sets: $S_{X_1, X_2}, S_{X_1}, S_{X_2}, \mbox{ and } S_f$ with $\mu_{X_1, X_2} = 10, \mu_{X_1} = \mu_{X_2} = 5, \mu_f = 5$ (b) the members of $S_{X_1, X_2}$ (c) binary representation of members of $S_{X_1, X_2}$ (d) the concatenated binary representation. If the string `$000010$' is drawn, then `$000$' is given to informant 1 and `$010$' is given to informant 2. (e) Function output values corresponding to $\overline{X}$ for `bitwise OR'.}
\label{fig:probExa_f}
\vspace{-0.2in}
\end{figure*}

\textit{Note on the terminology:} We call a bit-location in the bit-string at an informant (as well as in the bit-representation of $\overline{X}$ in encoding scheme defined above) \textit{defined}, if the sink knows its value unambiguously, otherwise it is called \textit{undefined}. For example, until the sink learns of the actual $\overline{X}$ revealed to the informants, one or more bits in the $\sum_{i = 1}^N \lceil \log \mu_{X_i} \rceil$ bits long representation of $\overline{X}$, remain \textit{undefined}. Similarly, a bit-location in the bit-representation of the output of the function $f$ is called \textit{evaluated} if the sink can unambiguously compute its value based on the values of one or more bits in informant strings.

\section{Communication Protocol for Distributed Function Computation}
\label{sec:commProtocol}
We address the distributed function computation problem, introduced in the last section, in \textit{bit-serial} communication scenarios, where in each communication round, only one informant can send only one bit to the sink. This is an example of scenarios where communication takes place over a channel with uplink throughput constrained to one bit per channel use. Our interest in this communication model stems from it allowing us to compute the minimum number of informant bits (total and individual) required to compute $f(\overline{X})$ at the sink when any number of rounds and sink bits can be used. In other words, this communication scenario enables us to compute the worst-case achievable rate-region for this problem, as we show later in this section.

We provide a constructive solution of the distributed function computation problem of the last section, based on interactive communication. The proposed protocol optimally solves this problem and computes the worst-case achievable rate-region. We call the proposed protocol ``\textbf{b}it-\textbf{ser}ial \textbf{f}unction \textbf{Comp}utation (\textbf{bSerfComp})'' protocol and describe it next.

\subsection{The \textbf{bSerfComp} protocol}
\label{subsec:bSerfCompAL}
In \textbf{bSerfComp} protocol, in each communication round only one bit is sent by the informant chosen to communicate with the sink. The chosen bit has the property that it divides the size of the current conditional ambiguity set of function output values, at the sink, closest to half\footnote{For $n$-ary representation of data-values, this will be $1/n$.}. Formally, in terms of the problem statement and encoding introduced in the last section, if $U$ is the set of \textit{undefined} bits in $\sum_{i = 1}^N \lceil \log \mu_{X_i} \rceil$ bits long representation of $\overline{X}$, then the bit chosen in $l^{\textrm{th}}, l \ge 0$, round is the one that solves $\argmin_{j \in U} \max_{b(j) \in \{0, 1\}} \mu^l_{f|b(j)}$. The sink, after receiving the value of the chosen bit, recomputes the set of undefined bits $U$. This is carried out iteratively till all bits in $\lceil \log \mu_f \rceil$ bits long representation of $f(\overline{X})$ are not \textit{evaluated}.

\hspace{-1.0em}\hrulefill

\hspace{-0.5em}{\textbf{Algorithm:} bSerfComp }

\vspace{-0.2cm}\hspace{-1.0em}\hrulefill
\begin{codebox}
\li $l = 0$
\li Let $S_{X_1,\ldots, X_N}^l = S_{X_1,\ldots, X_N}$
\li Let $S_f^l = S_f$, $\mu_f^l = |S_f^l|$
\li Let $V = \{1, \ldots, \sum_{i = 1}^{N} \lceil \log \mu_i \rceil\}$
\li Let $U$ be the set of undefined bits in $V$, $U \subseteq V$, over all $\overline{X} \in S_{X_1,\ldots, X_N}^l$
\li \While ($\mu_f^l > 1$)
\li 	$K^{l+1} = \argmin_{j \in U} \max_{b(j) \in \{0, 1\}} \mu^l_{f|b(j)}$
\li 	Choose the bit-location corresponding to $k^{l+1}$, where $k^{l+1}$ is a randomly chosen element of $K^{l+1}$
\li 	The sink asks the informant corresponding to bit-location $k^{l+1}$ to send the bit-value $b(k^{l+1})$ \label{bSerf:bitValue}
\li 	Set $S_{X_1,\ldots, X_N}^{l+1} = S_{{X_1,\ldots, X_N}|b(k^{l+1})}^l$
\li 	Set $S_f^{l+1} = S_{f|b(k^{l+1})}^l$
\li 	Compute $U \subset V$, the set of undefined bits
\li 	$l = l + 1$
    \End
\end{codebox}
\vspace{-0.2cm}\hrulefill

The sink can perform the worst-case performance analysis of the \textbf{bSerfComp} protocol by selecting on the line~\ref{bSerf:bitValue}, $b^*(k^{l+1})$ that solves:
\begin{equation*}
b^*(k^{l+1}) = \argmax\limits_{s = \{0, 1\}} \mu^l_{f|b(k^{l+1})=s}
\end{equation*}

Note that there are two versions of the \textbf{bSerfComp} protocol: in the \textit{online} version, the sequence of queries from the sink to the informants is determined adaptively depending on the informant response in the previous rounds, while in the \textit{offline} version, for a given support-set of data-vectors the entire sequence of queries is determined before actual querying starts. For example, the sequence of queries generated for the worst-case analysis of the protocol corresponds to the offline version.

\subsection{Optimality of \textbf{bSerfComp} protocol}
\label{subsec:bSerfCompOptimality}
The binary representations of the elements of $S_f$, as in Figure~\ref{fig:probExa_f}.e, can be arranged as the leaves of a binary tree, where ambiguity set of function output values $S_f$ forms the root and conditional ambiguity sets of function output values form internal nodes and leaves. The set of function output values corresponding to a child node is obtained by conditioning the set of function output values corresponding to its parent node on the value $b, b \in \{0, 1\}$ of $b^i_j$: $j^{\textrm{th}}$ bit-location in the binary string revealed to $i^{\textrm{th}}$ informant, with `$b = 0$' leading to the left subtree and `$b = 1$' leading to the right subtree. Such a binary tree with $\mu_f$ leaves will have a minimum-height of $\lceil \log \mu_f \rceil$, implying that at least $\lceil \log \mu_f \rceil$ bits are required to describe any leaf, in the worst-case.

\begin{pavikl}
\label{lemma:binTrees}
\textbf{bSerfComp} protocol computes all minimum-height binary trees corresponding to the given support-set to \textit{exactly} evaluate a given function $f$.
\end{pavikl}
\begin{IEEEproof}
Follows from the definition of minimum-height binary trees and the description of \textbf{bSerfComp} protocol.
\end{IEEEproof}

\begin{pavikl}
\label{lemma:minBits}
\textbf{bSerfComp} protocol computes $b_i$, the minimum number of bits that the $i^{\textrm{th}}, i \in \{1, \ldots, N\}$, informant must send to let the sink \textit{exactly} evaluate the function $f$.
\end{pavikl}
\begin{IEEEproof}
The \textbf{bSerfComp} protocol exploits the bit-serial communication scenario where a bit queried from the chosen informant maximally conditions the resultant ambiguity set of function output values at the sink. Also, to reduce the number of bits that an informant sends, the \textbf{bSerfComp} protocol can procrastinate querying the bits from the concerned informant until it can be postponed no more, thus maximally reducing the number of bits an informant sends. Combining these two observations, proves the lemma.
\end{IEEEproof}

\begin{pavikl}
\label{lemma:rateRegion}
For a given support-set, each corner point of the worst-case achievable rate-region for computing function $f$ corresponds to at least one minimum-height binary tree, with height $\#_f$.
\end{pavikl}
\begin{IEEEproof}
For the sake of contradiction, let us assume that there is a corner point of the worst-case achievable rate-region to which no minimum-height binary tree corresponds to. This implies that this corner point is outside the worst-case rate-region defined by the set of all the corner points visited by the set of minimum-height binary trees. This further implies that at this corner point at least one informant, say $i^{\textrm{th}}$, sends fewer bits than $b_i$ (defined in the statement of Lemma~\ref{lemma:minBits} above). However, this contradicts the definition of $b_i$, that it is the minimum number of bits $i^{\textrm{th}}$ informant needs to send to let the sink \textit{exactly} evaluate the function $f$. Thus, there cannot be any corner point outside the rate-region defined by the set of corner points corresponding to the set of all minimum-height binary trees, hence proving the lemma.
\end{IEEEproof}

\begin{pavikt}
\label{thrm:bSerfComp}
For a given support-set, \textbf{bSerfComp} protocol computes the worst-case achievable rate-region for function $f$.
\end{pavikt}
\begin{IEEEproof}
Combining the statements of Lemmas~\ref{lemma:binTrees} and \ref{lemma:rateRegion}, we can state that \textbf{bSerfComp} protocol computes each corner point of the worst-case achievable rate-region. Thus, \textbf{bSerfComp} protocol computes the worst-case achievable rate-region for computing function $f$.
\end{IEEEproof}

The worst-case achievable rate-region for distributed computation of function $f$ in asymmetric communication scenarios is given by the following corollary to Theorem~\ref{thrm:bSerfComp}. For the sake of notational simplicity, we state it only for $N=2$.
\begin{pavikc}
For $N = 2$, if $b_i$ denotes the minimum number of bits that an informant $i, 1 \le i \le 2$, sends over all solutions of \textbf{bSerfComp} protocol and $\#_f$ denotes the total number of bits sent by all informants, then the worst-case achievable rate region is given by:
\begin{eqnarray*}
R_1 & \ge & b_1 \\
R_2 & \ge & b_2 \\
R_1 + R_2 & \ge & \#_f
\end{eqnarray*}
\end{pavikc}
\begin{proof}
The proof follows from the worst-case optimality of \textbf{bSerfComp} protocol proven in Theorem~\ref{thrm:bSerfComp}.
\end{proof}

In Figures~\ref{fig:dist6fComp}-\ref{fig:dist5fComp}, using \textbf{bSerfComp} protocol, we compute the worst-case achievable rate-regions for functions: `bitwise OR', `bitwise AND', and `bitwise XOR', evaluated at sink over two support-sets of data-vectors for two correlated informants\footnote{In computer programming literature, it is well-know how to compute the bitwise functions over two binary strings, \cite{088k&r_book}. Let $B(X_i, X_j)$ denote the output of the bitwise function $B$ evaluated over binary-strings corresponding to $X_i$ and $X_j$, $i, j \in \{1, \ldots, N\}, i \neq j$. Define $B(X_i) = X_i$. Then, the evaluation of $B$ over any number $N, N \ge 2$, of arguments can be recursively defined, for example, as: $B(X_1, \ldots, X_N) = B(B(X_1, \ldots, X_{N-1}), X_N)$.}.

\subsection{Performance bounds for \textbf{bSerfComp} protocol}
\label{subsec:bSerfCompBounds}
To compute the bounds on the performance of \textbf{bSerfComp} protocol, we make use of an interesting and important observation regarding the working of \textbf{bSerfComp} protocol to compute a given function $f$ at the sink for a given support-set of data-vectors.

\textit{Observation:} A bit-location in the bit-representation of the function output values can be \textit{evaluated}, without all bit-locations in the concatenated bit-representation of $\overline{X}$ being \textit{defined}.

Let $\#_f$ and $\#_{DSC}$ denote the minimum number of informant bits required, in the worst-case, to evaluate the function $f$ and to solve the DSC problem, respectively, at the sink for a given support-set of data-vectors.

\textbf{\textit{Loose Bounds:}} As we discussed before, $\#_f$ is bounded from below by $\lceil \log \mu_f \rceil$, that is
\begin{equation*}
\lceil \log \mu_f \rceil \le \#_f
\end{equation*}

Let $\mu_s$ be the number of data-vectors or informant strings which evaluate to the function output value $s, s \in S_f$. Also, let us define $\mu^{min}_s = \min_{s \in S_f} \mu_s$. Then, assuming that the function $f$ evaluates to the output that corresponds to $\mu^{min}_s$, we obtain a trivial but useful upper bound on $\#_f$ as:
\begin{equation*}
\#_f \le \#_{DSC} - \lceil \log \mu^{min}_s \rceil
\end{equation*}

Therefore, combining above two bounds on $\#_f$, we can say that $\#_f$ loosely satisfies:
\begin{equation}
\label{eqn:looseBounds}
\lceil \log \mu_f \rceil \le \#_f \le \#_{DSC} - \lceil \log \mu^{min}_s \rceil
\end{equation}

\begin{figure}[!t]
\centering
\includegraphics[width=7.0in]{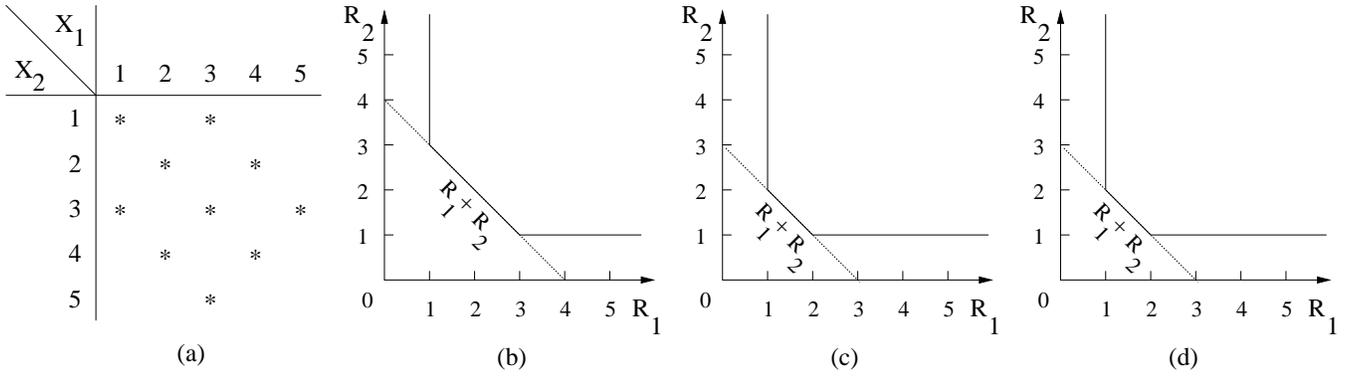}
\vspace{-0.1in}
\caption{Distributed function computation - I: support-set of data-vectors (a) and worst-case achievable rate-regions for `bitwise OR' in (b), for `bitwise AND' in (c), and for `bitwise XOR' in (d)}
\label{fig:dist6fComp}
\end{figure}

\begin{figure}[!t]
\centering
\includegraphics[width=7.0in]{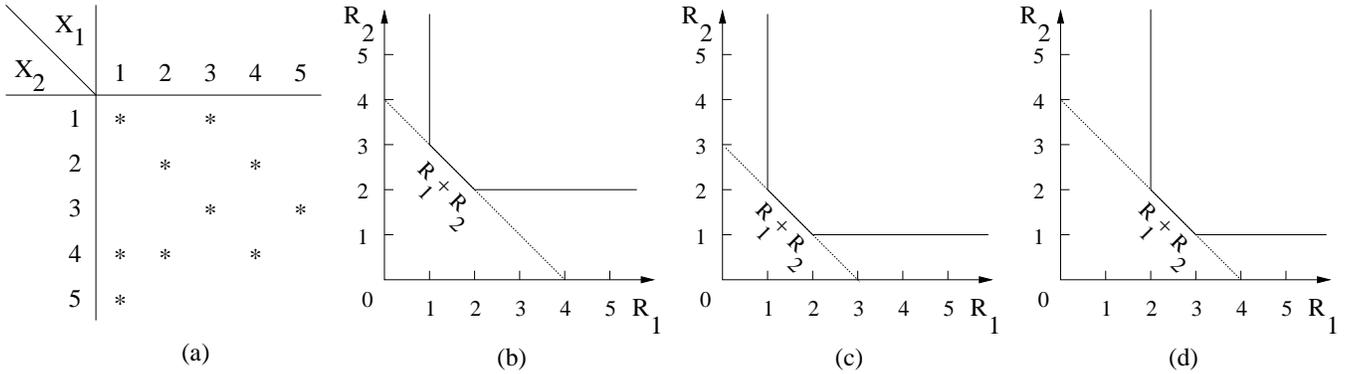}
\vspace{-0.1in}
\caption{Distributed function computation - II: support-set of data-vectors (a) and worst-case achievable rate-regions for `bitwise OR' in (b), for `bitwise AND' in (c), and for `bitwise XOR' in (d)}
\label{fig:dist5fComp}
\vspace{-0.2in}
\end{figure}

\textbf{\textit{Tight bounds:}} For a given support-set of data-vectors, $\lceil \log \mu_f \rceil$ is the lower bound on minimum number of informant bits required to evaluate the output of function $f$. Let us assume that the sink has obtained $\lceil \log \mu_f \rceil$ informant bits using \textbf{bSerfComp} protocol. Let assume that the size of conditional ambiguity set of data-vectors at the end of $l^{\textrm{th}}$ round, $1 \le l \le \lceil \log \mu_f \rceil$, is $1/2^{1-\epsilon_l}$ of its size at the beginning of this round. Define $\epsilon_{max} = \max \{\epsilon_1, \ldots, \epsilon_{\lceil \log \mu_f \rceil}\}$. Then, the size of conditional ambiguity set of data-vectors after $\lceil \log \mu_f \rceil$ informant bits are received satisfies:
\begin{equation*}
\frac{\mu_{X_1, \ldots, X_N}}{2^{\sum_{l=1}^{\lceil \log \mu_f \rceil} (1-\epsilon_l)}} \le \frac{\mu_{X_1, \ldots, X_N}}{2^{(1-\epsilon_{max}) {\lceil \log \mu_f \rceil}}}
\end{equation*}

Now, if
\begin{equation*}
\frac{\mu_{X_1, \ldots, X_N}}{2^{(1-\epsilon_{max}) {\lceil \log \mu_f \rceil}}} \le \mu^{min}_s
\end{equation*}
then, the function output evaluation finishes with $\lceil \log \mu_f \rceil$ to $\lceil \log \mu_f \rceil + \lceil \log \mu^{min}_s \rceil$ informant bits. So, we have
\begin{equation}
\label{eqn:tightBounds1}
\lceil \log \mu_f \rceil \le \#_f \le \lceil \log \mu_f \rceil + \lceil \log \mu^{min}_s \rceil
\end{equation}
and in this case the lower bound in \eqref{eqn:looseBounds} is actually tight.

Otherwise, that is, if
\begin{equation*}
\frac{\mu_{X_1, \ldots, X_N}}{2^{(1-\epsilon_{max}) {\lceil \log \mu_f \rceil}}} > \mu^{min}_s
\end{equation*}
then, the function computation finishes in $\lceil \log \mu_f \rceil + \lceil \log \mu^* \rceil$ to $\lceil \log \mu_f \rceil + \Big\lceil \log \frac{\mu_{X_1, \ldots, X_N}}{2^{(1-\epsilon_{max}) {\lceil \log \mu_f \rceil}}} \Big\rceil$ bits, where $\mu^*$ is the size of smallest subset of function output values that satisfies
\begin{equation*}
\frac{\mu_{X_1, \ldots, X_N}}{2^{(1-\epsilon_{max}) {\lceil \log \mu_f \rceil}}} \le \sum_{\stackrel{s \in S, S \subseteq S_f}{|S| = \mu^*}} \mu_s
\end{equation*}
Therefore, in this case we have:
\begin{equation}
\label{eqn:tightBounds2}
\lceil \log \mu_f \rceil + \lceil \log \mu^* \rceil \le \#_f \le \lceil \log \mu_f \rceil + \Big\lceil \log \frac{\mu_{X_1, \ldots, X_N}}{2^{(1-\epsilon_{max}) {\lceil \log \mu_f \rceil}}} \Big\rceil
\end{equation}

\section{Some Properties of Distributed Function Computation}
\label{sec:functionProperties}
We discuss some of the significant results, properties, and observations based on our work on distributed function computation problem in asymmetric communication scenarios.

\subsection{Two Classes of Functions}
\label{subsec:functionClasses}
Let us consider two deterministic functions $g = \max\{X_1, X_2\}$ and $h = X_1^{X_2}$. For two or more data-vectors derived from any discrete and finite support-set, the function $g$ may evaluate to same output value. On the other hand, the function $h$ assigns, in general, a unique output value to each of its input pairs. Generalizing this to the functions of $N, N \ge 2$, variables computed over corresponding discrete and and finite support-sets, there are various functions whose behavior is either like function $g$ or like function $h$ above.

The common statistical functions, such as `max', `min', `majority, `mean', `median', and `mode' and logical functions, such as `parity', `bitwise OR', `bitwise AND', and `bitwise XOR' belong to a class of functions, which we call \textit{lossy} functions. Similarly, the functions such as `identity function', `iterated exponentiation', $\sum^N_{i \neq j} X_i e^{X_j}$ belong to a class of functions, which we call \textit{lossless} functions. Formally, for the \textit{lossy} functions the cardinality of their range is smaller than the cardinality of their domain, while for the \textit{lossless} functions two cardinalities are equal. In fact, it is easy to prove that the equality of the sizes of domain and range of a function is an equivalence relation and classes of \textit{lossy} and \textit{lossless} functions are equivalence classes.

The reason these two equivalence classes of functions are relevant in the discussion of distributed function computation is that, in general, the computation of \textit{lossy} functions at the sink requires fewer number of informant bits than computation of \textit{lossless} functions. As DSC belongs to the equivalence class of \textit{lossless} functions (DSC is distributed function computation with function to be computed being the identity map: ${id}_{\overline{X}}$), this implies that, in general, for a given support-set the computation of \textit{lossless} functions requires as many informant bits, in the worst-case, as the solution of DSC problem, while the computation of \textit{lossy} functions requires fewer number of informant bits than DSC.

The \textbf{bSerfComp} protocol of the last section can be used to compute both, the \textit{lossy} and \textit{lossless} functions. However, as for the \textit{lossless} functions, the \textbf{bSerfComp} protocol reduces to much simpler \textbf{bSerCom} protocol of \cite{allerton08} for computing DSC in the corresponding communication scenario, the latter can be deployed at the sink for their computation in asymmetric communication scenarios.

Also, for the \textit{lossless} functions the knowledge of function output allows us to uniquely determine the input data-vector revealed to the informants (reversible function computation), while for the \textit{lossy} functions this is not possible (irreversible function computation). This \textit{apparent} loss of information accompanying the computation of \textit{lossy} functions results in their computation with fewer number of informant bits, but at the cost of sacrificing our ability to recover the input data-vector from their output. For the \textit{lossless} functions, there is no such information loss in their computation, allowing the unambiguous recovery of the input data-vectors from their output, but at the cost of larger number of informant bits.

\begin{table}[thb]
\centering
\caption{Comparison of \textit{lossy} and \textit{lossless} function computation}
\begin{tabular}{ l | l }
\hline
\textit{Lossy} Functions & \textit{Lossless} Functions  \\\hline
1. \textit{Examples:} various common statistical and bitwise functions & 1. \textit{Examples:} DSC, iterated exponentiation \\
2. Range of the function is smaller than its domain & 2. Range of the function is of same size as its domain \\
3. Complex \textbf{bSerfComp} protocol is used for computation & 3. Simple \textbf{bSerCom} of \cite{allerton08} is used for computation \\
4. The worst-case rate-region is larger than DSC & 4. The worst-case rate-region coincides with DSC \\
5. The sink cannot unambiguously recover the data-vector revealed & 5. The sink can unambiguously recover the data-vector revealed \\
\hspace{0.095in} to the informants from function output value & \hspace{0.095in} to the informants from function output value \\\hline
\end{tabular}
\label{table:fComparison}
\end{table}

It should be noted that above classification of functions holds true for any given support-set of data-vectors, in general. However, one can always concoct exceptions where the cardinality of the support-set of function output values for some \textit{lossy} function is same as the cardinality of the corresponding support-set of data-vectors. Similarly, some exception for \textit{lossless} functions can be constructed, where the cardinality of the support-set of function output values is smaller than the cardinality of corresponding support-set of data-vectors. We state without proof that the number of such instances of support-sets is small for any given cardinality of the support-sets. Further, in all situations the following lemma always holds for any function $f$.
\begin{pavikl}
\label{lemma:functionClasses}
If $\lceil \log \mu_f \rceil = \lceil \log \mu_{DSC} \rceil$, then $\#_f = \#_{DSC}$. Also, if $\lceil \log \mu_f \rceil < \lceil \log \mu_{DSC} \rceil$, then $\#_f \le \#_{DSC}$.
\end{pavikl}
\begin{IEEEproof}
Omitted for brevity.
\end{IEEEproof}

This brings us to relating our work on function classification with Han and Kobayashi's work along similar lines, \cite{087hanKobayashi}. We establish two equivalence classes of functions: \textit{lossy} and \textit{lossless}. Given that DSC problem belongs to the class of \textit{lossless} functions, the worst-case achievable rate-region of \textit{lossless} functions coincides with the worst-case rate-region of DSC problem, while for \textit{lossy} functions it is correspondingly larger. In \cite{087hanKobayashi} too, the authors have introduced such dichotomy of functions of correlated sources: for one class of functions the achievable rate-region coincides with Slepian-Wolf rate-region and for another class it does not. However, in spite of apparent similarities in results, there are some basic differences. First, we are interested in the worst-case information-theoretic analysis while authors in \cite{087hanKobayashi} are concerned with average-case analysis. Second, our results pertain to \textit{one-shot} function computation, while \cite{087hanKobayashi} deploys block-encoding.

It is interesting to ask if for a given communication scenario, we can always construct two or more equivalence classes of functions based on the communication cost of their computation. This appears to be a largely unexplored problem and a systematic answer that also unifies various previous attempts to classify functions based on their communication costs, as in \cite{105giridharKumar}, \cite{087hanKobayashi}, and this paper, warrants our attention. Further, proposed two classes of the functions can be further refined based on other finer details of the functions and we actually expect the classification structure to be richer than just the dichotomous classification in the paper. As of now, our own work in these directions is in preliminary stage and we propose to address these issues comprehensively in the near future.

\subsection{Dependence of $\#_f$ on $\mu_f$ and $\mu_{DSC}$}
\label{subsec:relations}

\begin{figure}[!t]
\centering
\includegraphics[width=7.0in]{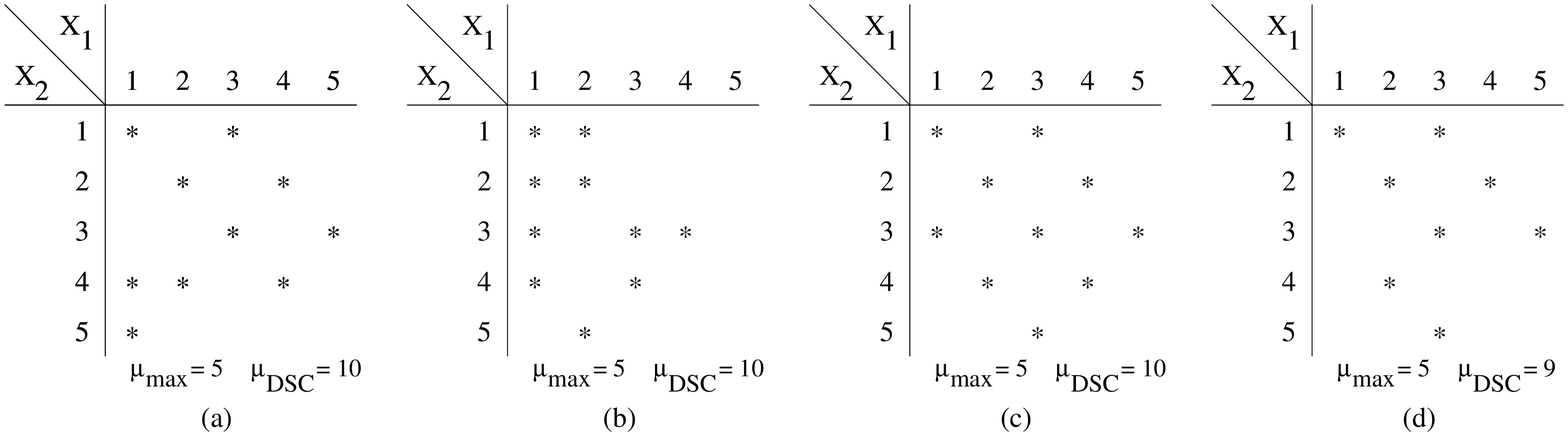}
\vspace{-0.1in}
\caption{`max' computation for support-set in (a) requires $\#_f = 4$ informant bits, (b) requires $\#_f = 3$ informant bits, (c) requires $\#_f = 4$ informant bits, and (d) requires $\#_f = 3$ informant bits}
\label{fig:maxComputation1}
\end{figure}

\begin{figure}[!t]
\centering
\includegraphics[width=7.0in]{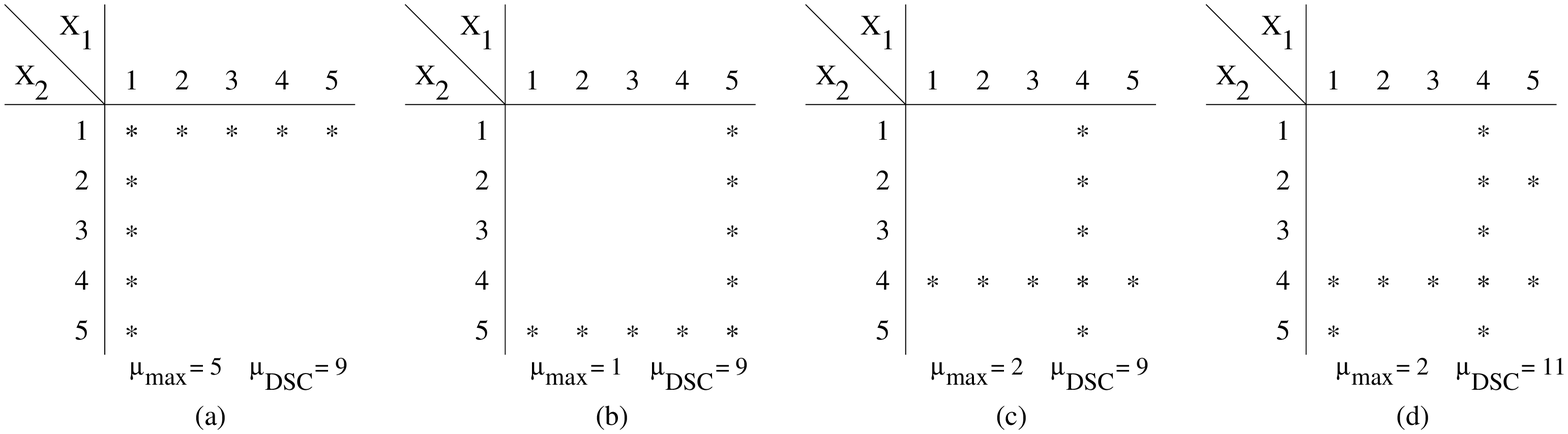}
\vspace{-0.1in}
\caption{`max' computation for support-set in (a) requires $\#_f = 6$ informant bits, (b) requires $\#_f = 0$ or no informant bits, (c) requires $\#_f = 2$ informant bits, and (d) requires $\#_f = 2$ informant bits}
\label{fig:maxComputation2}
\vspace{-0.2in}
\end{figure}

In subsection~\ref{subsec:bSerfCompBounds}, we established how for a given support-set of data-vectors $\#_f$, the minimum number of informant bits needed to compute the function $f$ in the worst-case, depends on $\mu_f$, the ambiguity of function output values, and $\mu_{X_1, \ldots, X_N}$, the ambiguity of data-vectors. Now, let us consider how for a given function $f$, $\#_f$ for two different support-set of data-vectors depends on corresponding $\mu_f$ and $\mu_{X_1, \ldots, X_N}$.

Let $\mu^1_f$ and $\mu^2_f$ denote the cardinality of the set of function output values for first and second support-set, respectively.

Let $\mu^1_{DSC}$ and $\mu^2_{DSC}$ denote the cardinality of the set of data-vectors for first and second support-set, respectively.

Finally, let $\#^1_f$ and $\#^2_f$ denote the minimum number of informant bits required to compute the function $f$ for first and second support-set, respectively.

In this subsection, we state without proof the relation between $\#^1_f$ and $\#^2_f$, given the relations between $\mu^1_f$ and $\mu^2_f$, and $\mu^1_{DSC}$ and $\mu^2_{DSC}$. We provide an exhaustive list of various possibilities and provide an example for each when the sink computes $\max\{X_1, X_2\}$ over data values of two informants for a given support-set.

\textit{Property 1:} $\mu^1_f = \mu^2_f$, $\mu^1_{DSC} = \mu^2_{DSC}$: $\#^1_f$ and $\#^2_f$ may or may not be equal, for example the support-sets in Figures~\ref{fig:maxComputation1}.(a)-(b) for which $\#^1_f \neq \#^2_f$. In this case, it is clear that ambiguities corresponding to function output values and data-vectors alone cannot be used to establish the relation between $\#^1_f$ and $\#^2_f$. This deficiency of the notion of ambiguity is addressed in greater detail in one of our related papers, \cite{paper_2}.

\textit{Property 2:} $\mu^1_f = \mu^2_f$, $\mu^1_{DSC} \neq \mu^2_{DSC} \implies \#^1_f$ and $\#^2_f$ follow the ordering of $\mu^1_{DSC}$ and $\mu^2_{DSC}$. An illustration of this case is given by the support-sets in Figures~\ref{fig:maxComputation1}.(c)-(d).

\textit{Property 3:} $\mu^1_f \neq \mu^2_f$, $\mu^1_{DSC} = \mu^2_{DSC} \implies \#^1_f$ and $\#^2_f$ follow the ordering of $\mu^1_{f}$ and $\mu^2_{f}$. Figures~\ref{fig:maxComputation2}.(a)-(b) illustrate this.

\textit{Property 4:} $\mu^1_f < \mu^2_f$, $\mu^1_{DSC} < \mu^2_{DSC} \implies \#^1_f \le \#^2_f$. Support-sets in Figures~\ref{fig:maxComputation2}.(c) and \ref{fig:maxComputation1}.(c) illustrate this.

\textit{Property 5:} $\mu^1_f < \mu^2_f$, $\mu^1_{DSC} > \mu^2_{DSC} \implies \#^1_f \le \#^2_f$. Support-sets in figures~\ref{fig:maxComputation2}.(d) and \ref{fig:maxComputation1}.(d) illustrate this.

\section{Conclusions and Future Work}
\label{sec:conclusions}
We address the distributed function computation problem in asymmetric and interactive communication scenarios, where the sink is interested in computing some deterministic function of input data that is split among $N$ correlated informants and is derived from some discrete and finite distribution. We consider the distributed function computation as a generalization of distributed source coding problem. We are mainly interested in computing $\#_f$, the minimum number of informant bits required \textit{in the worst-case}, to allow the sink to exactly compute the given function. We provide \textbf{bSerfComp} protocol to optimally compute the functions at sink for any given support-set of data-vectors and prove it computes the worst-case achievable rate-region for computing any given function, illustrating this with examples. Also, we provide a set of bounds on the performance of the proposed protocol.

We define two equivalence classes of functions: \textit{lossy} and \textit{lossless}. We show that the \textit{lossy} functions can be computed, in general, with fewer number of informant bits than \textit{lossless} function, such as DSC. Further, we establish the dependence of $\#_f$, when the function $f$ is computed over two different support-sets, on their respective ambiguities of function output values and data-vectors.

In future, we want to extend this work in three interesting directions. First, in this paper we have assumed that the sink and informants directly communicate with each other. Allowing the sink and informants to indirectly communicate with each other over one or more intermediate nodes (as in multihop networks), offers many more opportunities of reducing the number of bits carried over the network to compute a function at the sink. Second, allowing the sink to tolerate certain amount of error in the computation of the function may reduce the number of informant bits required. We want to address these directions formally in our setup. Finally, we want to come up with a generic framework to classify the functions based on the communication costs of their computation over arbitrary networks with any given model of communication and computation.

\end{document}